\documentclass[sigconf]{acmart}

\AtBeginDocument{%
  \providecommand\BibTeX{{%
    \normalfont B\kern-0.5em{\scshape i\kern-0.25em b}\kern-0.8em\TeX}}}

\graphicspath{{./images/}} 
\usepackage{algorithm}
\usepackage{algorithmic}
\usepackage{orcidlink}
\usepackage{hyperref}
\DeclareMathOperator{\polylog}{polylog}

\acmConference{}{}{}
\acmBooktitle{}
\acmISBN{}
\acmDOI{}
\acmYear{}

\renewcommand\footnotetextcopyrightpermission[1]{}
\setcopyright{none}

\begin{document}

%\title{Improved Analysis of Leader Election in Diameter-Two Networks}

\title{Toward Optimality: A Tighter Analysis of Message Complexity for Leader Election in Diameter-Two Networks}

% \author{Adri Bhattacharya}{Indian Institute of Technology Guwahati, Assam-781039, India}{a.bhattacharya@iitg.ac.in}{https://orcid.org/0000-0003-1517-8779}{}

\author{Abhijit Sadhukhan}
\email{abhijits3112@gmail.com}
\orcid{0009-0000-2947-4048} % Optional, but recommended by ACM
\affiliation{%
  \institution{Indian Statistical Institute}
  %\streetaddress{123 Research Blvd}
  \city{Kolkata}
  \country{India}
}

\author{Adri Bhattacharya }
\email{adri.bhattacharya@gmail.com}
\orcid{0000-0003-1517-8779} % Optional, but recommended by ACM
\affiliation{%
  \institution{Indian Statistical Institute}
  %\streetaddress{123 Research Blvd}
  \city{Kolkata}
  \country{India}
}

\author{Anisur Rahaman Molla}
\email{molla@isical.ac.in}
\orcid{0000-0002-1537-3462} % Optional, but recommended by ACM
\affiliation{%
  \institution{Indian Statistical Institute}
  %\streetaddress{123 Research Blvd}
  \city{Kolkata}
  \country{India}
}

\begin{abstract}
 % We study the message complexity of implicit leader election in synchronous diameter-two networks. Our main contribution is a tighter analysis of the randomized protocol proposed by Chatterjee et al. [\emph{Distributed Computing, 2020}]. In their work, they established a lower bound of $\Omega(n)$ messages and presented an $O(n \log^3 n)$ message-bound randomized algorithm for leader election in diameter two networks. The algorithm runs in $O(1)$ rounds and succeeds with high probability. In this paper, we significantly reduce the message complexity to $O(n\log n)$, by providing a refined analysis of their randomized algorithm. 

We study the message complexity of leader election in synchronous networks of diameter two. Our main contribution is a refined analysis of the randomized algorithm proposed by Chatterjee et al. [DC, 2020]. In their work, the authors established a lower bound of $\Omega(n)$ messages ($n$ is the number of nodes in the network) and presented a randomized algorithm that elects a leader in ${O}(1)$ rounds using $O(n \log^3 n)$ messages with high probability. 

In this paper, we improve their $\polylog n$ gap in the message bound by providing a tighter analysis of their algorithm, reducing the message complexity to $O(n\log n)$, while preserving the $O(1)$-round complexity and high-probability correctness guarantee.   
  %[\emph{Distributed Computing, 2020}]
   % We study the message complexity of implicit leader election in synchronous diameter-two networks. Our main contribution is a tighter analysis of the randomized protocol proposed by Chatterjee et al. [\emph{Distributed Computing, 2020}] . In their work, they established a lower bound of $\Omega(n)$ messages and presented a randomized Monte Carlo algorithm that runs in $O(1)$ rounds and succeeds with high probability. They showed that the algorithm requires $O(n \log^3 n)$ messages with high probability. We significantly reduce the message complexity to $O(n\log n)$, by providing a refined analysis of their randomized algorithm. 
\end{abstract}

%%
%% The code below is generated by the tool at: http://dl.acm.org/ccs.cfm

%\ccsdesc[500]{Theory of Computation~Distributed Algorithms}

\keywords{distributed algorithm, leader election, randomized algorithm, message complexity, diameter-two networks}

%% The following are not a requirement, delete if not using
% \received{20 February 2024}  %% inital submission date
% \received[revised]{12 March 2024} %% interim new draft
% \received[accepted]{5 June 2024}  %% publication version

\maketitle

\section{Introduction and Related Works}\label{sec:intro}
Leader election is a fundamental problem in distributed computing. In this problem, nodes in the network collectively attempt to elect a single node as their leader. Leader election serves as a building block for many other distributed tasks and has been extensively studied across a wide variety of network topologies \cite{humblet1984electing4, korach1990modular6, peleg1990time17, santoro2006design19, tel2000introduction20}. In the implicit version of leader election \cite{lynch1996distributed14}, the elected node knows that it is the leader, while all other nodes know that they are not elected. The problem was initially studied mainly in the deterministic setting \cite{afek1991time1, santoro2006design19, tel2000introduction20}. Recent work has explored randomized solutions and established several strong results in the randomized setting.  Kutten et al. \cite{kutten2015sublinear12} presented a near-optimal sublinear algorithm for leader election in complete graphs (i.e., diameter one), which runs in $O(1)$ round and uses $O(\sqrt n \log ^{3/2} n)$ messages with high probability and also showed $\Omega(\sqrt{n})$ message lower bound. In general graphs with diameter three or more, Kutten et al. \cite{kutten2015complexity11} showed a lower bound of $\Omega(m)$ messages for randomized leader election and also developed an algorithm which uses $O(m\log\log n)$ messages with high probability. Gilbert et al. \cite{gilbert2018leader2} also studied randomized algorithms for implicit leader election in well-connected general graphs and analyzed complexity bounds in terms of mixing time. For additional details and comparisons, we refer the reader to \cite{kutten2015sublinear12,kutten2015complexity11}. 

% Later, Gilbert et al. \cite{gilbert2018leader2} studied randomized algorithms for implicit leader election in well-connected general graphs. In \cite{kutten2015complexity11}, Kutten et al. proved a lower bound of $\Omega(m)$ messages for leader election in general graphs with diameter at least three, where $m$ is the total number of edges in the underlying graph. Subsequently, Kutten et al.  \cite{kutten2015sublinear12} presented a near-optimal sublinear algorithm for leader election in complete graphs. 

For diameter-two networks, Chatterjee et al. \cite{chatterjee2020complexity} established a lower bound of $\Omega(n)$ messages ($n$ is the number of nodes in the network) together with a simple near-optimal Monte Carlo algorithm, completing the understanding of the problem up to polylog gaps between upper and lower bounds. 
Their proposed Monte Carlo algorithm that works in $O(1)$ rounds and uses $O(n \log^3 n)$ messages with high probability\footnote{i.e., with probability at least $1 - \frac{1}{n^c}$ for some constant $c\geq 1$.} and does not assume knowledge of the network size $n$ or any other global parameter knowledge. 

% Our primary goal is to design leader election algorithms that are efficient in terms of both time and communication (message) complexity. In \cite{chatterjee2020complexity}, the authors proved a lower bound of $\Omega(n)$ on the message complexity for diameter-two networks and proposed a Monte Carlo algorithm that works in $O(1)$ rounds and uses $O(n \log^3 n)$ messages with high probability (i.e., with probability at least $1 - \frac{1}{n^c}$ for some constant $c$). Their algorithm does not assume knowledge of the network size $n$, and therefore does not rely on any global information.

\subsection*{Our Contributions:}
%\textbf{Need more refined version and add some line on what types of different cases we study, in terms of which we break the cases etc.}
In this work, we provide a tighter analysis of the randomized algorithm developed by Chatterjee et al.  \cite{chatterjee2020complexity}. For the sake of completeness, we include the pseudocode in Algorithm~\ref{alg:degree_based_leader_election}. Our analysis improves the message complexity bound from $O(n \log^3 n)$ to $O(n \log n)$ with high probability. This significant improvement nearly closes the gap between the upper bound and the lower bound of $\Omega (n)$. Our analysis follows the same high-level framework as the original work, including the use of degree-based bucketing and concentration bounds via the Chernoff inequality. 

Our improvement comes from a more refined case analysis based on the expected contribution of each degree bucket to the total message complexity. Instead of applying loose upper bounds at intermediate steps, we delay aggregation and bound the relevant sums only at the end. This finer control over intermediate terms allows us to obtain a significantly tighter overall concentration bound.

\section{Model and Problem Definition}
The underlying synchronous congest model follows that of Chatterjee et al.~\cite{chatterjee2020complexity} and has also been used in 
\cite{lynch1996distributed14, afek1991time1, humblet1984electing4, korach1990modular6, korach1985optimality8, korach1989optimal9}. 
The communication network is a connected undirected graph $G=(V,E)$ with $|V|=n$ and $|E|=m$. Each node (processor) in $G$ executes an instance of a distributed algorithm in synchronous rounds. 
In every round, a node may send messages to its neighbors, receive messages sent in the same round, and perform arbitrary local computation. 
Messages are of size $O(\log n)$ bits and constitute the only means of communication; nodes do not share memory. 
Each node has a unique identifier of size $O(\log n)$ bits. 
All nodes are initially awake and start executing the algorithm simultaneously. 
This is the \emph{KT0 model}~\cite{peleg1990time17}, where nodes initially know only their own identities. 

In this paper, we focus on graphs with diameter $D(G)=2$. 
Since $G$ has diameter two, we have $n-1 \le m < \frac{n(n-1)}{2}$. A formal definition of the leader election problem is given below.
\begin{definition}[Leader Election]\label{def:le}
Each node $u$ maintains a variable $\mathit{status}_u$ taking values from $\{\perp, \textsc{ELECTED}, \textsc{NON-ELECTED}\}$
with $\mathit{status}_u = ~\perp$ initially. 
An algorithm $\mathcal{A}$ solves the leader election problem in $T$ rounds if, from round $T$ onward, exactly one node sets its status to \textsc{ELECTED} and every other node sets its status to \textsc{NON-ELECTED}. 
\end{definition}
This is the standard implicit version of the leader election problem. In the explicit version, we additionally require that every non-leader node knows the identity of the leader node. In this work we discuss on implicit leader election in diameter two network.

%  *************
% The communication network is represented by an undirected connected graph 
% $G = (V,E)$ with $|V| = n$ with  diameter two. 

% Computation proceeds in synchronous rounds. In each round, every node performs 
% some local computation, receives and sends (possibly different) messages to its neighbors. We assume that each 
% message has size of at most $O(\log n)$ bits.

% The time complexity is measured in the number of rounds, and the message 
% complexity is the total number of messages exchanged during the execution. 
% An event holds with high probability (w.h.p.) if it occurs with probability 
% at least $1 - 1/n^c$ for some constant $c > 0$.

%\textbf{Need to write}

%\paragraph{Related work}  1. Briefly state the contain of our source paper. 2. state the result available for this problem in other networks. (citations will be here mostly)

%\paragraph{Our contribution} Clearly state what is precisely our improvement complexity wise. 

%\section{Previous work}  1. Clearly state the algorithm as it was.  2. State the main theorems of their complexity bounds and explain the key points so that out contribution became easily understandable. (In this part we state correctness theorems from the previous paper, we do not repeat this correctness analysis) 3. Along with those explanation briefly add our contribution there. We may add a paragraph on our contribution (briefly explain)

 \section{Message Complexity Analysis}
%here we concretely do the analysis of concentration bounds.

%We mention briefly that we already stated the correctness and obvious statement for constant time.

%we only do the full message complexity analysis.

%We have already mentioned that we don't give a new algorithm, we tightening the analysis. We also mentioned that the algorithm runs in $O(1)$ rounds and succeed with high probability (\textbf{Mention the corresponding theorem}). In this section we give the full analysis of message complexity and will show that the algorithm uses $O(n \log n)$ messages with high probability.To make things easier to understand we use the same notation as of the previous paper. Let $X$ be a random variable denotes the total number of candidate selected and $X_v$ be a indicator random variable taje value $1$ if $v$ became a candidate. Again let $M^{entire}$ be another random variable which denote the total number of message in he the algorithm and $M_v$ be the number of messages sent and received by node $v$.

\begin{algorithm}[t]
\caption{Randomized Leader Election~\cite{chatterjee2020complexity}}
\label{alg:degree_based_leader_election}
\begin{algorithmic}[1]

\STATE Each node $v \in V$ independently becomes a candidate with probability$\frac{1 + \log d_v}{d_v},$
where $d_v$ is the degree of node $v$.

\IF{$v$ becomes a candidate}
    \STATE $v$ sends its ID to all its neighbors.
\ENDIF

\STATE Each node acts as a referee for all candidate neighbors (including itself, if applicable).

\STATE Upon receiving IDs from candidate neighbors $v_1, v_2, \dots, v_j$, a node $w$ computes $\min \{ \text{ID}(v_1), \text{ID}(v_2), \dots, \text{ID}(v_j) \}.$

\STATE Node $w$ sends this minimum ID back to each of $v_1, v_2, \dots, v_j$.

\STATE Each candidate node $v$ declares itself the leader if and only if it receives its own ID from all its neighbors.

\STATE Otherwise, $v$ declares itself a non-leader.

\end{algorithmic}
\end{algorithm}

We present a tighter analysis of the randomized Monte Carlo leader election algorithm (cf. Algorithm \ref{alg:degree_based_leader_election}). The algorithm runs in $O(1)$ rounds and
succeeds with high probability (cf. Lemma~3 in ~\cite{chatterjee2020complexity}). In this section, we prove the following theorem.

\begin{theorem}\label{main:theorem}
    Algorithm \ref{alg:degree_based_leader_election} succeeds to elect a leader in $O(1)$ rounds with high probability, while sending $O(n\log n)$ messages. 
\end{theorem}

We basically provide a
complete analysis of the message complexity and show that Algorithm \ref{alg:degree_based_leader_election} uses
$O(n \log n)$ messages with high probability.

For consistency, we follow the notation of~\cite{chatterjee2020complexity}.
Let $G = (V,E)$ be an undirected graph with $|V| = n$. For each node $v \in V$,
let $d_v$ denote its degree.

\medskip

Each node independently becomes a candidate with probability $p_v = \frac{1+\log d_v}{d_v}.$
Let $X_v$ be the indicator random variable defined as
\[
X_v =
\begin{cases}
1 & \text{if node $v$ becomes a candidate},\\
0 & \text{otherwise.}
\end{cases}
\]

Let $X = \sum\limits_{v \in V} X_v$ denote the total number of candidates selected.

%\medskip

Let $M_v$ denote the total number of messages sent and received by node $v$
during the execution of the algorithm, and let $M^{\mathrm{entire}} = \sum\limits_{v \in V} M_v$ denote the total number of messages used by Algorithm \ref{alg:degree_based_leader_election}.

\subsection*{Expectation of the Total Number of Messages:}

By the structure of the algorithm, if a node $v$ becomes a candidate,
it communicates with all of its neighbors. Thus, $M_v = 2 d_v \cdot X_v.$

Therefore, $M^{\mathrm{entire}} = \sum\limits_{v \in V} 2 d_v X_v.$

Taking expectation and using linearity of expectation,
\[
\mathbb{E}\!\left[M^{\mathrm{entire}}\right]
= \sum_{v \in V} 2 d_v \mathbb{E}[X_v].
\]

Since $X_v$ is an indicator random variable,
\[
\mathbb{E}[X_v] = p_v = \frac{1+\log d_v}{d_v}.
\]

Hence, $\mathbb{E}\!\left[M^{\mathrm{entire}}\right]
= \sum\limits_{v \in V} 2 d_v \cdot \frac{1+\log d_v}{d_v}
= 2 \sum\limits_{v \in V} (1+\log d_v).$

Now, since $d_v \le n$ for every $v \in V$, we have $\log d_v \le \log n$.

Therefore, $\sum\limits_{v \in V} (1+\log d_v)
\le \sum\limits_{v \in V} (1+\log n)
= n(1+\log n).$

Thus, $\mathbb{E}\!\left[M^{\mathrm{entire}}\right]
\le 2n(1+\log n)
= O(n \log n).$

\subsection*{Concentration via Bucketing:}

To obtain a high-probability bound on $M^{\mathrm{entire}}$, we follow the
bucketing approach of~\cite{chatterjee2020complexity}. Since the variables
$M_v$ are not identically distributed (they depend on $d_v$), we group nodes
according to their degrees.

Let, $k$ be an integer such that $2^{k-1} \le n < 2^k$. For each integer $i$ such that, $1 \le i \le k $, define the bucket $V_i = \{ v \in V : 2^{i-1} \le d_v < 2^{i} \}.$ and let $n_i = |V_i|$ denote the number of nodes in bucket $V_i$.

Again assume, $Y_i = \sum_{v \in V_i} X_v$ denote the number of candidates selected from bucket $V_i$.

\medskip

We now bound the expectation of $Y_i$.

\begin{lemma}
\label{bound of Y_i}
For every $i \ge 2$, $\mathbb{E}[Y_i] \le \frac{3 i n_i}{2^i}.$
\end{lemma}

\begin{proof}
By definition, 
\[\mathbb{E}[Y_i]
= \sum_{v \in V_i} \mathbb{E}[X_v]
= \sum_{v \in V_i} \frac{1+\log d_v}{d_v}.
\]

For any $v \in V_i$, we have $2^{i-1} \le d_v < 2^{i}.$

Hence, $\log d_v \le \log(2^{i}) = i,$
and $\frac{1}{d_v} \le \frac{1}{2^{i-1}}.$

Therefore, 
\[
\frac{1+\log d_v}{d_v}
\le \frac{1+i}{2^{i-1}}
\le \frac{3i}{2^i}.
\]

The above relation holds for $i\ge 2$.
%for all sufficiently large $i$ (and the bound can be absorbed into constants for small $i$).

Thus, $\mathbb{E}[Y_i]
\le \sum\limits_{v \in V_i} \frac{3i}{2^i}
= \frac{3i \cdot n_i}{2^i} ,$
as required.
\end{proof}

\paragraph{\textbf{High-Probability Bound on Message Complexity:}}

This is our main technical section, where we prove that this algorithm uses $O(n \log n)$ messages with high probability.
%\begin{theorem}[Chernoff Bound {\cite{chatterjee2020complexity}}]
%Let $Z_1, Z_2, \ldots, Z_m$ be independent indicator random variables and
%let
%\[
%Z = \sum_{j=1}^{m} Z_j,
%\qquad \mu = \mathbb{E}[Z].
%\]
%Then for any $R \ge 6\mu$,
%\[
%\Pr[Z \ge R] \le 2^{-R}.
%\]
%\end{theorem}

\begin{theorem}
%There exists a constant $C > 0$ such that
For large enough $n,$
\[
\Pr\!\left[ M^{\mathrm{entire}} \ge O(n \log n) \right]
\le \frac{1}{n^4}.
\]

\end{theorem}

\begin{proof}

Recall that for each bucket $V_i$,
\[
Y_i = \sum_{v \in V_i} X_v
\quad \text{and} \quad
M_i = \sum_{v \in V_i} M_v.
\]

Since $M_v = 2 d_v X_v$, we have $M_i = \sum\limits_{v \in V_i} 2 d_v X_v.$

Recall that, for every $v \in V_i$, we have $2^{i-1} \le d_v < 2^{i}$.
Thus,
\[
M_i \le \sum_{v \in V_i} 2 \cdot 2^{i} X_v
= 2^{i+1} \sum_{v \in V_i} X_v
= 2^{i+1} Y_i.
\]

Hence, $M_i \le 2^{i+1} Y_i$. 

We have already established that
$\mathbb{E}[Y_i] \le \frac{3 i n_i}{2^i},$ for $i\ge2$ (cf. Lemma \ref{bound of Y_i}).

\medskip

Now we consider three cases based on the bucket index $i$, classified according to the expected number of candidates selected from that bucket. In all cases we will show that the messages used by all the buckets are bounded by $O(n \log n)$ with high probability. 

Here we separate the Case-I for $i=1$, because, in this proof, we use Lemma \ref{bound of Y_i} which holds for $i\ge 2$. 
\begin{description}
\item[Case I:] $i=1$

Since $2^{i-1} \le d_v < 2^i$, here $d_v=1$ and hence clearly total message used by this bucket is, 
\[
M_1 =2 d_v \cdot \sum_{v \in V_1} X_v \le2 d_v \cdot n_1 \le 2d_v\cdot n=2n
\]
as $d_v=1$ and $n_1 \le n$. 

%We separate this case from $i\ge2$ for applying the Lemma \ref{bound of Y_i}, which holds for $i\ge 2$. 

\item[Case II:] $i\ge2$ and $\mathbb{E}[Y_i] \ge \log n$
%\vspace{1 cm}

Let, $R = 6 \cdot \frac{3 i n_i}{2^i}.$

So, we have
$R \ge 6 \, \mathbb{E}[Y_i]$, since $\mathbb{E}[Y_i]\le \frac{3in_i}{2^i}$. 

Since $Y_i$ is a sum of independent indicator random variables, so we can use Chernoff bound (\cite[Theorem 4.4(3)]{lynch1996distributed14}) and hence we get,
\begin{align*}
&\Pr[ Y_i \ge R ]
\le 2^{-R}\\ 
\Rightarrow &\Pr[ Y_i \ge \frac{6 \cdot3in_i}{2^i} ] \le 2^{\frac{-6 \cdot3in_i}{2^i}}.
\end{align*}
% \[
% \Pr[ Y_i \ge R ]
% \le 2^{-R} 
% \Rightarrow \Pr[ Y_i \ge R ] \le 2^{\frac{-6 \cdot3in_i}{2^i}}.
% \]

Since, $\frac{3in_i}{2^i} \ge \mathbb{E}[Y_i] \ge \log n$ in this case,
\[
\Pr\!\left[ Y_i \ge\,\frac{ 6\cdot 3in_i}{2^i} \right]
\le2^{\frac{-6 \cdot3in_i}{2^i}}\le 2^{-6\log n}
= \frac{1}{n^6}.
\]

Now, $M_i \le 2^{i+1} Y_i$ leads to,
\[
\Pr\!\left[
M_i \ge  \frac{6 \cdot 3 i n_i}{2^{i}} \cdot 2^{i+1}
\right]
\le \frac{1}{n^6},
\]

that is, $\Pr\!\left[
M_i \ge 36 i n_i
\right]
\le \frac{1}{n^6}.$

\medskip

Let, us consider the following set $\mathcal{I} = \{ i : \mathbb{E}[Y_i] \ge \log n \}$. Applying a union bound over all such buckets, and noting that the number
of buckets is at most $\log n$, we obtain
\[
\Pr\!\left[
\sum_{i \in \mathcal{I}} M_i
\ge
\sum_{i \in \mathcal{I}} 36 i n_i
\right]
\le
\frac{\log n}{n^6}
< \frac{1}{n^5}.
\]

Observe that, 
\[
\sum\limits_{i \in \mathcal{I}} i n_i
\le
\sum\limits _{i = 2}^k i n_i \le \sum\limits _{i = 2}^k \log n \cdot n_i \le \log n \sum\limits _{i = 2}^k n_i \le n \log n, 
\]
as we have, $\sum\limits_{i\ge1}n_i=n$ and, $2^{k-1} \le n < 2^k$.

We conclude, $\Pr\!\left[
\sum\limits_{i \in \mathcal{I}} M_i
\ge 36 n \log n
\right]
\le \frac{1}{n^5}.$

\medskip

Let, $M' = \sum\limits_{i \in \mathcal{I}} M_i.$ Then, $\Pr\!\left[ M' \ge 36 n \log n \right]
\le \frac{1}{n^5}.$

\item[Case III:] $i\ge2$ and $\mathbb{E}[Y_i] < \log n$

Let $R = 6 \log n.$

Since $\mathbb{E}[Y_i] < \log n$, we have, $R \ge 6 \mathbb{E}[Y_i].$

Applying the Chernoff bound ,
\[
\Pr\!\left[ Y_i \ge 6 \log n \right]
\le 2^{-6\log n}
= \frac{1}{n^6}.
\]

Recall that $M_i \le 2^{i+1} Y_i$. Therefore,
\[
\Pr\!\left[ M_i \ge 6 \cdot 2^{i+1} \log n \right]
\le \frac{1}{n^6}.
\]

\medskip
Consider the following set $\mathcal{J} = \{ i : \mathbb{E}[Y_i] < \log n \}$ and denote $M'' = \sum\limits_{i \in \mathcal{J}} M_i.$

Using a union bound over all such buckets (at most $\log n$ buckets),
we obtain
\[
\Pr\!\left[
M'' \ge
6 \log n \sum_{i \in \mathcal{J}} 2^{i+1}
\right]
\le
\frac{\log n}{n^6}
\le \frac{1}{n^5}.
\]

\medskip

Now, observe that the bucket indices satisfy the following inequalities: $1 \le i \le k$, and $2^{k-1} < n \le 2^k$.

Hence, $\sum\limits_{i=1}^{k} 2^{i+1}
\le 2^{k+2}
\le 8n.$ 

Therefore, $\Pr\!\left[
M'' \ge 48 n \log n
\right]
\le \frac{1}{n^5}.$

\end{description}
%\medskip

Combining all cases, we have the total message complexity as follows.
\[
M' = \sum_{i \in \mathcal{I}} M_i,
\qquad
M'' = \sum_{i \in \mathcal{J}} M_i,
\qquad
M^{\mathrm{entire}} = M_1+ M' + M''.
\]

From Case I, II, III, applying a union bound over these three events,
\[
\Pr\!\left[
M^{\mathrm{entire}} \ge 85 n \log n
\right]
\le \frac{3}{n^5}
\le\frac{1}{n^4},
\]
for large enough $n$.

Hence we conclude, $$\Pr\!\left[ M^{\mathrm{entire}} \ge O(n \log n) \right]
\le \frac{1}{n^4},$$

which completes the prove of Theorem \ref{main:theorem}.
\end{proof}

\section{Conclusion} 

We revisited the randomized leader election algorithm of Chatterjee et al.~\cite{chatterjee2020complexity} for diameter-two networks and provided a tighter analysis of its message complexity from $O(n \log^3 n)$ to $O(n \log n)$ with high probability. Since the known lower bound is $\Omega(n)$, this narrows the gap between upper and lower bounds from a $\log^3 n$ factor to a single $\log n$ factor. Closing this remaining $O(n)$ versus $O(n \log n)$ gap remains a challenging open problem.

%\textcolor{red}{Better to add a .bib file where you add all the references. Already there exists a .bib file, modify these below references and paste it there, then cite in the paper}

%clarify the compression of gaps and the remaining gaps.

\bibliographystyle{ACM-Reference-Format}
\bibliography{references.bib}
%\begin{thebibliography}{9}

% \bibitem{KuttenUniversal}
% S. Kutten, G. Pandurangan, D. Peleg, P. Robinson, and A. Trehan.
% \newblock On the complexity of universal leader election.
% \newblock {\em Journal of the ACM}, 62(1):7:1--7:27, 2015.
% \newblock DOI: 10.1145/2699440.

% \bibitem{KuttenSublinear}
% S. Kutten, G. Pandurangan, D. Peleg, P. Robinson, and A. Trehan.
% \newblock Sublinear bounds for randomized leader election.
% \newblock {\em Theoretical Computer Science}, 561(Part B):134--143, 2015.
% \newblock DOI: 10.1016/j.tcs.2014.02.009.

% \bibitem{ChatterjeeLeaderElectionDiameterTwo}
% S. Chatterjee, G. Pandurangan, and P. Robinson.
% \newblock The complexity of leader election in diameter-two networks.
% \newblock {\em Distributed Computing}, 33:189--205, 2020.
% \newblock DOI: 10.1007/s00446-019-00354-2.

% \end{thebibliography}

%%
\end{document}